\documentclass[a4paper,conference,10pt,twocolumns]{IEEEtran}
\usepackage{amsmath}
\usepackage{amsthm}
\usepackage{array}
\usepackage{epsf}
\usepackage{epsfig}
\usepackage{subfigure}
\usepackage{cite}
\usepackage{times}
\usepackage{textcomp}
\usepackage{graphicx}
\usepackage[utf8]{inputenc}
\usepackage{relsize}
\usepackage{amssymb}
\usepackage{flushend}

\newtheorem{lemma}{Lemma}

\begin{document}

\title{Censoring for Improved Sensing Performance in Infrastructure-less Cognitive Radio Networks}


\author{\IEEEauthorblockN{Mohamed Seif\IEEEauthorrefmark{1},
Mohammed Karmoose\IEEEauthorrefmark{2}, and Moustafa Youssef \IEEEauthorrefmark{3}}
\IEEEauthorblockA{\IEEEauthorrefmark{1}Wireless Intelligent Networks Center (WINC), Nile University, Egypt\\
\IEEEauthorrefmark{2} Department of Electrical Engineering, Alexandria University, Egypt\\
\IEEEauthorrefmark{3}  Department of Computer Science, Egypt-Japan University of Science and Technology (E-JUST), Egypt\\
Email: m.seif@nu.edu.eg,
mkarmoose@ucla.edu, 
moustafa.youssef@ejust.edu.eg}}


\maketitle

\begin{abstract}
Censoring has been proposed to be utilized in wireless distributed detection networks with a fusion center to enhance network performance in terms of error probability in addition to the well-established energy saving gains. In this paper, we further examine the employment of censoring in infrastructure-less cognitive radio networks, where nodes employ binary consensus algorithms to take global decisions regarding a binary hypothesis test without a fusion center to coordinate such a process. We show analytically - and verify by simulations - that censoring enhances the performance of such networks in terms of error probability and convergence times. Our protocol shows performance gains up to 46.6\% in terms of average error probability over its conventional counterpart, in addition to performance gains of about 48.7\% in terms of average energy expenditure and savings up to 50\% in incurred transmission overhead.
\end{abstract}

\IEEEpeerreviewmaketitle

\section{Introduction}
Cognitive Radios Networks (CRNs) have emerged as a viable solution to the problem of inefficient spectrum utilization under the current spectrum licensing paradigm\cite{fcc}. In the opportunistic spectrum access cognitive model, secondary users (SUs) that do not possess a license to use a spectrum band are nevertheless allowed to transmit whenever the licensed or primary users (PUs) are not active. Spectrum sensing is therefore a mandatory task for SUs to detect the presence of the PUs in order to identify the available transmission opportunities\cite{haykin2005cognitive}. Due to the fact that spectrum sensing by a single SU may be highly unreliable, cooperative sensing can be employed such that the decision regarding the presence or absence of PUs is based on measurements taken by a cluster of SUs, thereby enhancing the reliability of the taken decisions. However, this comes at an inevitable cost of increased transmission overhead and energy expenditure.

To lessen the impact of this cost, \textit{censoring} was first introduced by Rago \cite{rago1996censoring} in wireless sensor networks with soft-decision detection frameworks where sensors send the locally-computed Log-Likelihood Ratio (LLR) values to the fusion center (FC), which is responsible for making a global decision. Energy savings were attained from employing censoring in this framework on the expense of loss in terms of average error probability. Recently, censoring was considered in a hard-decision framework, where sensors apply one-bit quantization to their local measurements prior to transmission. In \cite{seddik2011censoring,karmoose2012censoring,karmoose2014performance}, it was proven that censoring enhances the performance of such networks in terms of error probability of the global decision, in addition to the previously attained and characterized energy savings. Specifically, in \cite{seddik2011censoring} and \cite{karmoose2012censoring}, networks with parallel topologies were considered with Time-Division-Multiple-Access (TDMA) and Type-Based-Multiple-Access (TBMA) techniques respectively, while in \cite{karmoose2014performance} networks with sequential topologies were considered.

However, all these mentioned censoring systems have the common assumption that a fusion center is present for coordination among the transmitting SUs, global decision-making, and finally reporting back the decision to SUs. Assuming the presence of a fusion center in fact contradicts with the current paradigm of large and/or ad hoc CRNs\cite{youssef2013routing,karmoose2013dead}. In addition, as was stated in \cite{mishra2006cooperative}, reporting the sensing information from a large number of SUs to the fusion center may be problematic. Moreover, in \cite{sun2007cluster}, the authors show the deteriorating impact of multipath fading in the reporting channel on the detection performance of a centralized cooperative sensing framework.

The major contribution of this paper is to propose a censoring-based hard-decision distributed detection framework that is well-suited for \emph{\textbf{infrastructure-less}} CRNs (i.e., does not contain a designated central coordinator). Our proposed framework does not require the presence of a fusion center for data collection and coordination. Instead, we propose the use of binary consensus algorithms which allow SUs to exchange binary information regarding the presence or absence of a PU with direct neighbors and eventually arrive at a global decision based on the collective decisions of direct neighbors. The closest work to ours is \cite{olfati2007consensus,ashrafi2011binary}, where binary consensus algorithms are used for distributed spectrum sensing in a non-censoring fashion. However, we show analytically - and verify through numerical simulations - that our proposed infrastructure-less censoring-based system significantly outperforms conventional systems in terms of average error probability (up to 46.6\%), overall incurred overhead (up to 50\%), and energy expenditure of SUs (up to 48.7\%).

The rest of the paper is organized as follows. Section~\ref{binary_consensus_algorithm} gives a background review on binary consensus algorithms and its application in CRNs. Section~\ref{sec:model} describes the system model. We analyze the system in Section~IV and validate the obtained expressions in Section~V. Section~VI provides numerical evaluation of the attained expressions and shows the performance gains for our proposed system. Finally, we conclude in Section~VII.

\section{Background: Binary Consensus Algorithm}
\label{binary_consensus_algorithm}
In this section, we provide a basic review of binary consensus algorithms in distributed spectrum sensing\cite{olfati2007consensus}. To allow SUs to cooperatively arrive at a global decision and with no help from a designated central entity, each SU makes a local decision regarding the presence or absence of a PU, denoted as $H_1$ or $H_0$ respectively. SUs then exchange their binary local decisions with their direct neighbors for $K$ time steps, where $K$ is the running time of the algorithm. Upon the termination of the algorithm, each SU individually makes a decision for $H_1$ or $H_0$, based on the received decisions from neighboring nodes. Let $b(k)=[b_{1}(k),\dots,b_{M}(k)]^{T}$ be the vector of local decisions at time step $k$ at $M$ SUs. The binary consensus algorithms are summarized as follows:
\begin{enumerate}
 \item At first time step ($k=0$), each SU initially transmits its local decision to the neighbors it is connected to at this time step.
 \item At each consecutive time step ($0 < k < K$), each SU collects the decisions transmitted by neighboring SUs. It then combines these decisions, along with the previously received decisions from past time steps, through a \textit{combining function} which generates a new decision $b(k)$ to be transmitted to neighboring nodes at the current time step $k$. This can be mathematically expressed as:
 \begin{equation}
 \label{combining_function}
  b(k) = \mathcal{F}(b(n), \; n=0,\cdots,k-1) \: 0<k<K-1
 \end{equation}

 \noindent where $\mathcal{F}(.)$ is the combining function.
 \item Upon the termination of the algorithm ($k=K$), each node makes a final decision based on the previously obtained decisions from all time steps $0<k<K$, through a \textit{decision function}. This is mathematically expressed as:
  \begin{equation}
 \label{decision_function}
  b(K) = \mathcal{D}(b(n), \; n=0,\cdots,K-1)
 \end{equation}
  \noindent where $\mathcal{D}(.)$ is the decision function.
\end{enumerate}

Based on the appropriate choice of the combining and decision functions, the binary consensus algorithm is guaranteed to converge to a common decision after a sufficiently long running time, i.e., $b_{i}(K) = b^\star, \; \forall i=1,\cdots,M$ as $K \rightarrow \infty$\cite{olfati2007consensus}.

In this paper, we focus on one variation of binary consensus algorithms, i.e., \textit{Diversity-based binary consensus algorithm}\cite{ashrafi2011binary}, in which a SU uses its initial local decision for decision reporting at all consecutive time steps, and the combining function is basically a majority rule for the received decisions along the $K$ time steps. The combining and decision functions in this case are mathematically expressed as:

\vspace{- 0.1 in} 

\begin{equation}
\begin{split}
 \mathcal{F}(b(n), \; n=0,\cdots,&k-1) = b(k-1), \hspace{20pt} 1< k <K, \\
 \mathcal{D}(b(n), \; n=0,\cdots,&K-1) = \\
 &\text{Dec}\bigg(\dfrac{1}{M}(b(0)+\dfrac{1}{Kp}\sum\limits_{t=0}^{K-1} \boldsymbol{A}(t)b(t))\bigg)
\end{split}
\end{equation}

\noindent where $\text{Dec}(x)=\begin{cases}
                                      1, & \text{if } x \geq 0 \\
                                      0, & \text{if } x < 0
                                     \end{cases}$
\noindent with 1 and 0 corresponding to deciding $H_1$ and $H_0$, respectively. When $x$ is a vector, the function operates on it element-wise.

\begin{table}
\footnotesize
 \begin{tabular}{|c|l|}
 \hline
  \textbf{Symbol} &\textbf{Description} \\
  \hline
  $M$ & Number of SUs. \\
  $\bar{\tau}_{ij}(k)$ & Instantaneous SNR of $i$th SU at $j$th SU at time step k. \\
  $\tau$ & Min. acceptable threshold for successful decoding. \\
  $f_s$ & Center frequency of the energy detector Bandpass filter. \\
  $B$ & Bandwidth of the energy detector Bandpass filter. \\
  $T$ & Time interval of the integrator in the energy detector. \\
  $r_{i}(t)$ & Received primary signal by $i$th SU at time $t$. \\
  $s(t)$ & Signal transmitted by the PU. \\
  $n_{i}(t)$ & Additive white Gaussian noise. \\
  $h_{i}(t)$ & Complex-Gaussian channel gain between the PU and $i$th SU. \\
  $x_{i}(0)$ & Output of the energy detector of the $i$th SU. \\
  $b_i(0)$ & Initial binary decision of $i$th SU. \\
  $\eta$ & Local detection threshold in the conventional case. \\
  $\eta_1$/$\eta_0$ & Upper/Lower detection thresholds in the censoring case. \\
  \hline
 \end{tabular}
\label{table::list_symbols}
\caption{List of symbols.}
\end{table}

\section{System Model}\label{sec:model}
In this section, we describe the system model of our proposed framework. TABLE~I enlists the various symbols used in the upcoming analysis.
\subsection*{Network Model}
We consider a CRN that consists of $M$ SUs that opportunistically transmit in the presence of a PU. We model the secondary network as an undirected random graph $\mathcal{G}=(\mathcal{N}, \mathcal{E})$, where $\mathcal{N}$, the set of nodes, represents the SUs, and $\mathcal{E}$, the set of edges, denotes the connectivity of SUs. A node $i$ is connected to node $j$ if $\bar{\tau}_{ij}(k) > \tau$, where $\bar{\tau}_{ij}(k)$ is the instantaneous SNR of the signal of SU $i$ at SU $j$ at time step index $k$, and $\tau$ is the minimum acceptable SNR required for successful decoding of secondary transmission. Assuming channel reciprocity, then $\bar{\tau}_{ji}(k) = \bar{\tau}_{ij}(k)$ and both SUs are in the neighborhood of each other if their instantaneous SNR exceeds the decoding threshold. The probabilistic nature of the wireless channel and therefore the instantaneous SNR of received secondary signals are the reasons behind the ``randomness'' in the graph $\mathcal{G}$. Due to the absence of a central entity for coordinating transmissions, nodes which are in the same
transmission range employ the CSMA/CA multiple
access protocol to concurrently use the bandwidth for transmission while reducing the probability of data collision.

\subsection*{Spectrum Sensing}
Each SU is equipped with an energy detector as shown in Figure \ref{fig::energy_detection}. A bandpass filter selects the desired primary channel with a center frequency $f_{s}$ and bandwidth $B$. The filter is followed by a squaring-law device and an integrator with time interval $T$ to measure the signal power. The output of the integrator for the $i$th SU is denoted by $x_{i}(0)$. We assume a flat-fading channel model between the PU and the SUs. Let $r_{i}(t)$ represent the received primary signal received by the $i$th SU. Let $H_1$ and $H_0$ represent the hypothesis of the presence or absence of a PU, respectively. The received signal can then be expressed as:

 \[ r_{i}(t) = \left\{
  \begin{array}{l l}
    n_{i}(t) & \quad \text{Given $H_{0}$}\\
    h_{i}(t)s(t)+n_{i}(t) & \quad \text{Given $H_{1}$}
  \end{array} \right.\]

\noindent where $s(t)$ is the signal transmitted by the PU, $n_{i}(t)$ is AWGN at the $i$th SU, and $h_{i}(t)$ is the complex-Gaussian channel gain between the PU and the $i$th SU.

\begin{figure}
  \includegraphics[height=2.8cm, width=9.1cm]{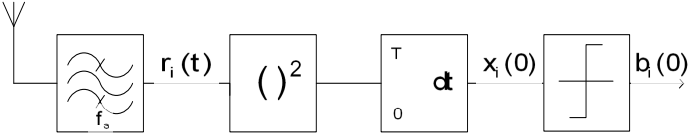}
\vspace{-1em}
\caption{Block diagram of an energy detector at each sensor node.}
\label{fig::energy_detection}
\end{figure}

\subsection*{Local Decisions}
\subsubsection{Conventional Case}
After obtaining a local measurement $x_i(0)$, each node makes a local decision regarding the presence or absence of the PU. The decision making process can be expressed as follows:

\begin{equation}
 b_i(0) = \left\{ \begin{array}{ll}
                     1, & \hbox{$x_{i}(0) \geq \eta$} \\
                     -1, & \hbox{$x_{i}(0) < \eta$} \\
                    \end{array}\right.
\end{equation}

\noindent  where $b_i(0)$ is the initial binary decision of node $i$ and $\eta$ is the local decision threshold of all SUs. A decision that is equal to $1$ and $-1$ corresponds to locally detecting $H_1$ and $H_0$, respectively.

\subsubsection{Censoring Case}
We allow nodes to \textit{censor} transmission: a node employs a two-threshold decision making process as shown in Figure \ref{fig::thresholds}, and withholds transmission in case a measurement falls in between the two thresholds. Let $\eta_1$ and $\eta_0$ denote the upper and lower thresholds respectively, then the detection process can be expressed as:

\begin{equation}
 b_i(0) = \left\{
 \begin{array}{ll}
            1,  & \hbox{$x_i(0) \geq \eta_1$} \\
            0, & \hbox{$\eta_0 \leq x_i(0) < \eta_1$} \\
            -1, & \hbox{$x_i(0) < \eta_0$}
            \end{array}\right.
\end{equation}

\noindent where $b_i(0)=0$ represents censoring.

\subsection*{Global Decision}
SUs employ the diversity-based binary consensus algorithm described in Section \ref{binary_consensus_algorithm} to exchange local decisions with direct neighbors and eventually arrive at a collective decision regarding the PU. To characterize the detection performance of our distributed system, we assume a global-AND rule of the decisions made by all SUs, i.e.,

\begin{equation}
\label{definition_pd_pfa}
 \begin{split}
  &P_d(K) \doteq \text{Pr}(b_{i}(K)=0, \forall i=1,...,M |H_1), \\
  &P_{fa}(K) \doteq \text{Pr}(b_{i}(K)=0, \forall i=1,...,M |H_0)
 \end{split}
\end{equation}
\noindent and the average error probability is defined as $P_e(K) \doteq \pi_{H_0}P_{fa}(K) + \pi_{H_1}(1 - P_{d}(K))$, where $\pi_{H_0}$ and $\pi_{H_1}$ are the prior probabilities of having $H_0$ and $H_1$ respectively. We adopt a global-AND rule since it is the most sensitive to detection error and therefore the most conservative of all other decision rules in terms of detection performance. This can be viewed as a worst-case for the primary network.
%

\section{System Analysis}
In this section, we analyze our proposed system and derive analytical expressions to characterize different performance metrics such as average error probability, energy expenditure and network overhead. We also provide numerical simulations for the performance metrics to ensure the validity of our obtained expressions.

\subsection{Average Error Probability}

Let $\gamma_i$ be the signal-to-noise ratio (SNR) of the received primary signal at the $i$th node. Given the energy detector model mentioned previously, and assuming identical statistical characteristics of $r_{i}(t), \forall i=1,\cdots,M$, it is found that the probability distribution of $x_i(0), \forall i=1,\cdots,M$ is\cite{digham2007energy}:

\[ x_{i}(0)|\gamma_{i} \sim  \left\{
  \begin{array}{l l}

 \chi^{2}_{2TB} & \quad \text{Given $H_{o}$}\\
    \chi^{2}_{2TB}(2\gamma_{i}) & \quad \text{Given $H_{1}$}
  \end{array} \right.\]

\noindent where $\chi^{2}_{2TB}$ and $\chi^{2}_{2TB}(2\gamma_{i})$ are central and non-central with $2\gamma_i$ non-centrality parameter chi-squared distributions, respectively, both with $2TB$ degrees of freedom.

For the conventional system, the local decision probabilities of the $i$th node under hypothesis $H_{1}$ and $H_0$, denoted by $\pi_{11}$ and $\pi_{10}$ respectively, can be formulated as \cite{ashrafi2011binary}:


\begin{figure*}
\vbox{
 \begin{equation}
 \label{pd_censoring}
 \small
\begin{split}
P_{d}(K) &\approx \sum\limits_{c=0}^{M} \sum\limits_{n}^{} \frac{M!}{n!\, c!\, (M-n-c)!}
\pi_{11}^{n} \pi_{01}^{c} \pi_{-11}^{M-n-c} \bigg[1-Q\bigg(\frac{n}{\sqrt{\frac{(1-p) (M-c)}{pK}}}\bigg)\bigg]^{c}
\bigg[1-Q\bigg(\frac{n}{\sqrt{\frac{(1-p) (M-1-c)}{pK}}}\bigg)\bigg]^{M-c},
\end{split}
\end{equation}}

\vbox{
\begin{equation}
\small
\label{pfa_censoring}
\begin{split}
P_{fa}(K) &\approx \sum\limits_{c=0}^{M} \sum\limits_{n}^{}\frac{M!}{n!c!(M-n-c)!}
\pi_{10}^{n}\pi_{00}^c\pi_{-10}^{M-n-c} \bigg[1-Q\bigg(\frac{n}{\sqrt{\frac{(1-p) (M-c)}{pK}}}\bigg)\bigg]^{c}
\bigg[1-Q\bigg(\frac{n}{\sqrt{\frac{(1-p) (M-1-c)}{pK}}}\bigg)\bigg]^{M-c}.
\end{split}
\end{equation}}

\end{figure*}

\begin{equation}
\label{pi_11}
 \begin{split}
\pi_{11}&\doteq \text{Pr}(b_{i}(0)=1|H_{1})={e}^{-\frac{\eta}{2}}\sum\limits_{n=0}^{TB-2}\frac{1}{n!}+{(\frac{\bar{\gamma}+1}{\bar{\gamma}})}^{TB-1} \\
&\times \bigg({e}^{-\frac{\eta}{2(\bar{\gamma}+1)}}-{e}^{-\frac{\eta}{2}}\sum\limits_{n=0}^{TB-2}\frac{1}{n!}\big(\frac{\eta \bar{\gamma}}{2(\bar{\gamma}+1)}\big)^{n} \bigg)
 \end{split}
\end{equation}

\begin{equation}
\label{pi_10}
 \begin{split}
\pi_{10} &\doteq \text{pr}(b_{i}(0)=1|H_{0}) = 1-\frac{\Gamma_{l}(TB,\frac{\eta}{2})}{\Gamma(TB)}
\end{split}
\end{equation}

\noindent where $\Gamma(x)$ and $\Gamma_{l}(x,y) $ are the Gamma and lower incomplete Gamma functions, respectively.

For the censoring system, we define the local decision probabilities of the $i$th node under hypothesis $H_j, j=0,1$ as:

\begin{equation}
 \begin{split}
  &\pi_{1j} = \text{Pr}(b_i(0)=1|H_j) = \text{Pr}(x_{i}(0)  >  \eta_{1}|H_{j}) \\
  &\pi_{-1j} = \text{Pr}(b_i(0)=-1|H_j) = \text{Pr}(x_{i}(0)  \leq  \eta_{0}|H_{j})\\
  &\pi_{0j} = \text{Pr}(b_i(0)=0|H_j) = 1- (\pi^c_{1j} + \pi^c_{-1j}).
 \end{split}
\end{equation}

It can be shown that $\pi_{11}$ and $1-\pi_{-11}$ have the same functional form as (\ref{pi_11}), and $\pi_{10}$ and $1-\pi_{-10}$ have the same functional form as (\ref{pi_10}), while substituting $\eta$ with $\eta_1$ and $\eta_0$ respectively in both cases.

To model the connectivity between SUs among the network, the \textit{adjacency matrix} $\boldsymbol{A}(k)$ is defined as:

\begin{equation}
a_{ij}(k) = \begin{cases}
1 &  \text{if $\bar{\tau}_{ij}(k) >= \tau, \: i \neq j$}  \\
0 & \text{otherwise}
\end{cases}
\end{equation}

\noindent where $a_{ij}(k), \; i,j \in \{1, \; \cdots \;,M \}$ denotes the $(i,j)$th element of the matrix $\boldsymbol{A}(k)$. $a_{ij}(k)=1$ in this context means that nodes $i$ and $j$ are connected, and vice versa. For ease of exposition, we neglect the intricate details of the wireless channel transmission and communication scheme and assume $a_{ij}(k), \; i \neq j, \; \forall \; k \geq 0$ to be Bernoulli random processes with $p=\text{Pr}(a_{ij}(k))=\text{Pr}(\bar{\tau}_{ji}(k)>\tau)$, which implicitly models wireless channel characteristics.

The following lemma provides the detection and false alarm probabilities for the censoring system.

\begin{lemma}
 Assume a cognitive radio network that consists of $M$ SUs employing a binary consensus algorithm as described in Section \ref{binary_consensus_algorithm}, where $p$ is the probability of having a reliable link between two SUs. Then for a sufficiently large $M$, the probability of detection and false alarm of the global decision can be approximated by (\ref{pd_censoring}) and (\ref{pfa_censoring}) respectively.
\end{lemma}

\begin{proof}
According to (\ref{definition_pd_pfa}), the probability of detection of the global decision can be expressed as:
\begin{equation}
\small
\label{prf}
\begin{split}
 P_d(K) &= \text{Pr}(b_{i}(K)=0, \forall i=1,...,M |H_1) \\
 &= \text{Pr}(y_{i}(K)>0, \forall i=1,...,M |H_1) \\
 &\overset { (a) }{ = } \sum\limits_{c=0}^{M} \sum\limits_{S \in \mathcal{S}} \text{Pr}(y_{i}(K)>0, \forall i=1,...,M |C,S,H_1) \\
 &\times \text{Pr}(C,S|H_1) \\
 &\overset { (b) }{\approx} \prod_{i=1}^{M}\mathcal{Q}\left( \dfrac{-\mu_{y_i}}{\sigma_{y_i}} \right) \times \text{Pr}(C,S|H_1)
 \end{split}
\end{equation}
\noindent where $y(K) = [y_{1}(K),\cdots,y_{M}(K)]^T =  \dfrac{1}{M}(b(0)+\dfrac{1}{Kp}\sum\limits_{t=0}^{K-1} \boldsymbol{A}(t)b(t))$, $S=b(0)^T\bold{1}$ and $\bold{1} = [1,\cdots,1]^T$, $C$ is the number of nodes having 0 as initial decisions, and $Q(z) = \dfrac{1}{\sqrt{2\pi}}\int_z^{\infty} \text{exp}(-t^2/2)dt$. The third equality (step (a)) comes from Bayesian chain rule, where $\mathcal{S}=\{-(M-c),-(M-c+2),\dots,M-c\}$. The fourth near-equality (step (b)) holds because for a sufficiently large $M$, we can apply the Central Limit Theorem (CLT) to approximate the value of $P_d(K)$. Accordingly, the distribution of $y_{i}$ given $S$ and $C$ is a Gaussian distribution with mean and variance that are equal to:

\begin{equation}
\label{parameters}
\begin{split}
\mu_{y_i}(K)&=S, \\
\sigma_{y_i}^{2}(K)&=\frac{1-P}{KP} |M-b_{i}^{2}(0)-C|
\end{split}
\end{equation}

For the second term in the third equality in (\ref{prf}), $\text{Pr}(C,S|H_1)$ can be viewed as a multinomial distribution\cite{forbes2011statistical} of having $C$ SUs deciding 0 (w.p. $\pi_{01}$) and $n = \left(S + M - C\right)/2$ SUs deciding a 1 (w.p. $\pi_{11}$). Therefore, it can be expressed as:

\begin{equation}
\label{multinomial}
\begin{split}
 \text{Pr}(C&,\; S|H_1) = \\
 \text{Pr}(C&=c,\; n = \dfrac{S + M - c}{2}|H_1) =\\
 & \dfrac{M!}{n! \; c! \;  (M-n-c)!} \times \pi_{11}^n \pi_{01}^{c} \pi_{-11}^{(M - n - c)}
 \end{split}
\end{equation}

Plugging (\ref{parameters}) and (\ref{multinomial}) in (\ref{prf}) gives (\ref{pd_censoring}). Similar approach can be followed to prove (\ref{pfa_censoring}).

\end{proof}

For comparison, similar expressions are obtained for the conventional (non-censoring) case based the results in \cite{ashrafi2011binary} with minor variations to account for polar transmissions. Those are:

\begin{equation}
\label{pd_conventional}
\begin{split}
 P_{d}(K)&\approx \sum\limits_{n=-M,-M+2,\dots}^{M}{M \choose \frac{M+n}{2}} \pi_{11}^{\frac{M+n}{2}} (1-\pi_{11})^{\frac{M-n}{2}} \\
 &\times \bigg[1-Q\bigg(\frac{n}{\sqrt{\frac{(1-p) (M-1)}{pK}}}\bigg)\bigg]^{M}
 \end{split}
\end{equation}

\begin{equation}
\label{pfa_conventional}
\begin{split}
P_{fa}(K) &\approx\sum_{n=-M,-M+2,\dots}^{M}{M \choose \frac{M+n}{2}} \pi_{10}^{\frac{M+n}{2}} (1-\pi_{10})^{\frac{M-n}{2}} \\
&\times \bigg[1-Q\bigg(\frac{n}{\sqrt{\frac{(1-p) (M-1)}{pK}}}\bigg)\bigg]^{M}
 \end{split}
\end{equation}

\begin{figure}
\centering
 \includegraphics[width=2.6in]{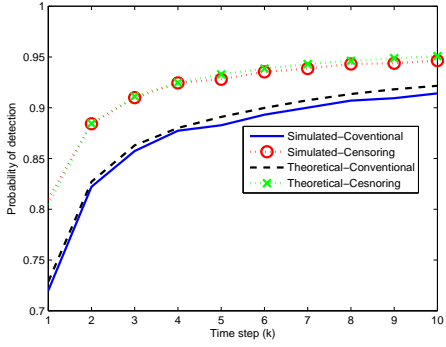}
\caption{Detection Probability: Simulation versus analytical expression - $M=51$, $p=0.8$, $\bar{\gamma}=2 dB$, $\eta=10.3$, $\eta_1=14.6$, $\eta_0=7$, $TB=5$.}
\label{fig::detection_verify}
\end{figure}

\begin{figure*}
 \centering
 \begin{minipage}{0.3\textwidth}
  \includegraphics[width=2.4in]{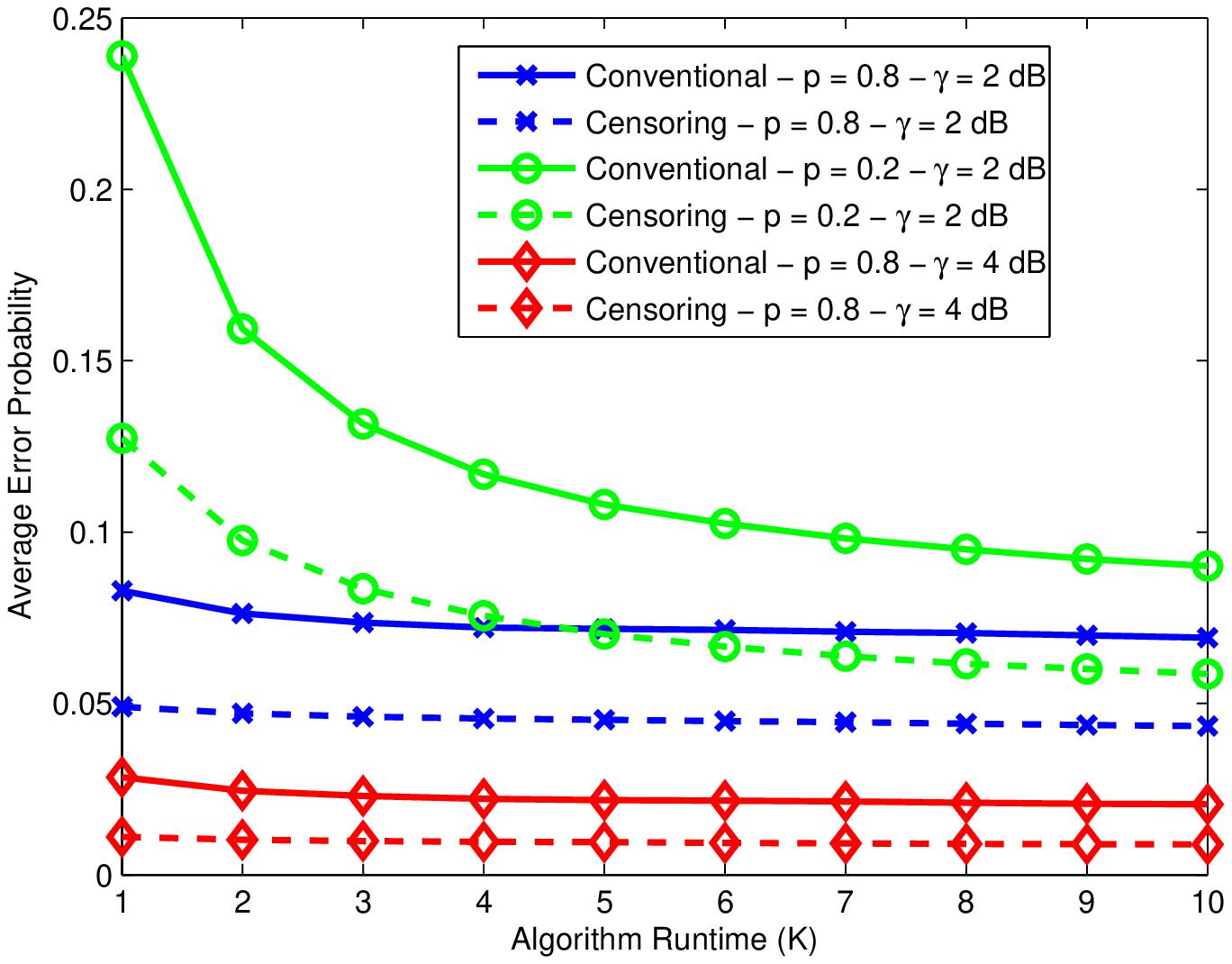}
\caption{Average probability of error (analytical) - $M=51$, $TB=5$.}
\label{fig::error_probability}
 \end{minipage}
  \begin{minipage}{0.3\textwidth}
  \begin{center}
\includegraphics[width=2.4in]{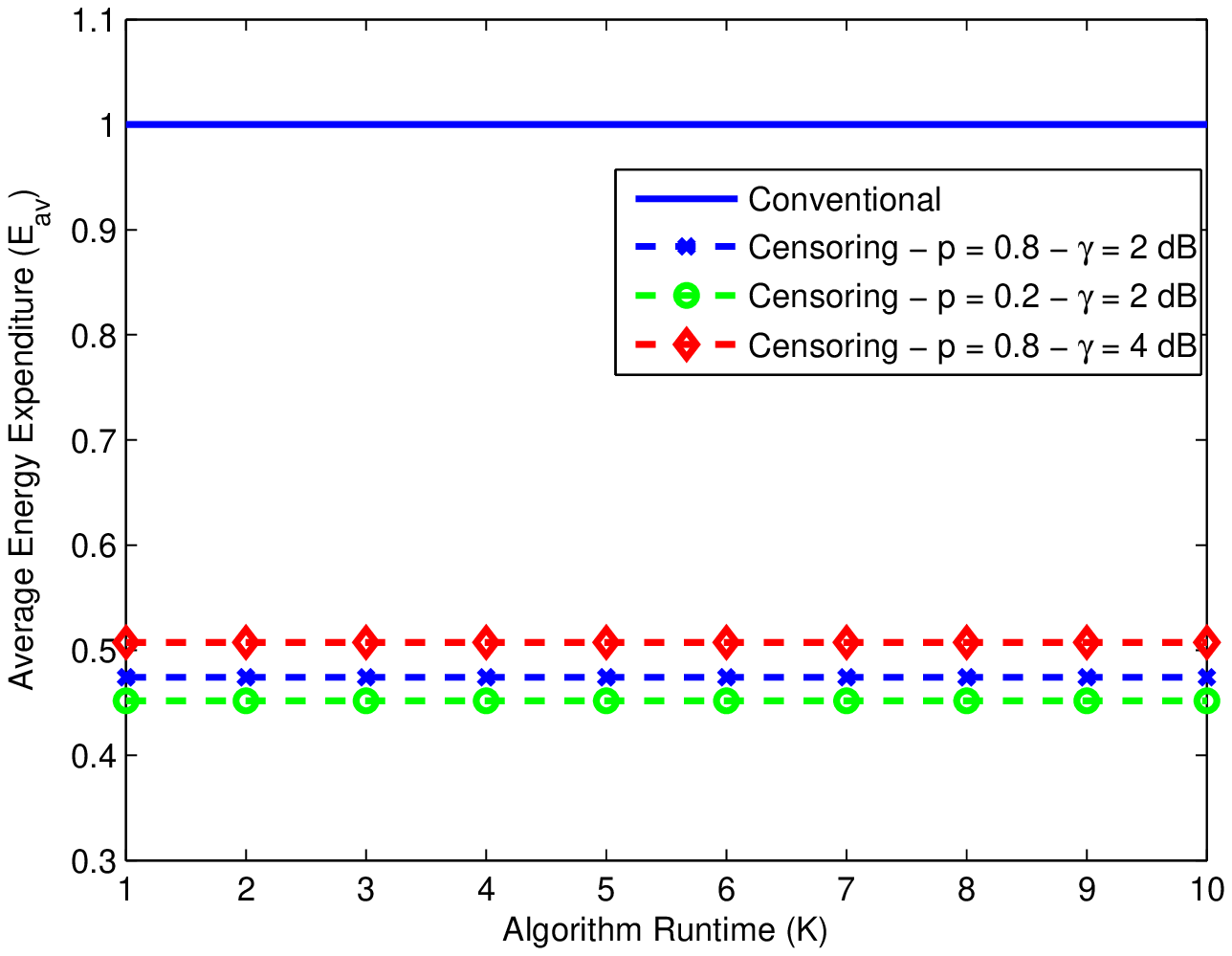}
\caption{ Energy expenditure (analytical) - $M=51$, $TB=5$.}
\label{fig::energy_expenditure}
\end{center}
 \end{minipage}
 \begin{minipage}{0.3\textwidth}
   \includegraphics[width=2.4in]{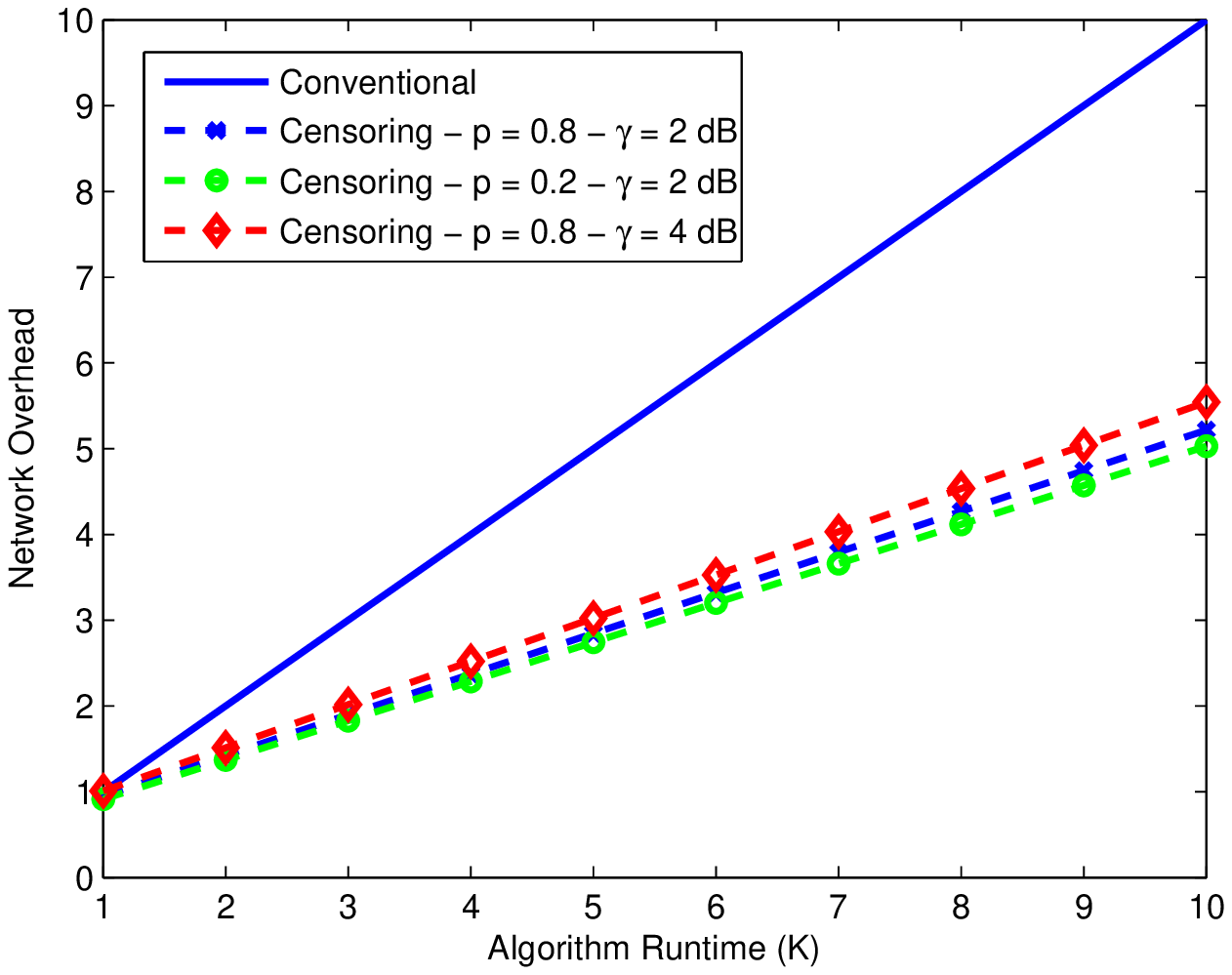}
\caption{Network overhead (analytical) - $M=51$, $TB=5$.}
\label{fig::overhead}
 \end{minipage}
\end{figure*}

\subsection{Energy Expenditure}
In a conventional system a node typically transmits an equal amount of energy for both decisions (assuming polar representation). Let $E$ be the energy consumed by a node to transmit a decision (either a 1 or a -1) to neighboring nodes; then the average energy consumed by a node is $E$. In a censoring-enabled system, however, a node saves the transmit energy during the censoring phase. If $E$ is the energy required to transmit a 1 or a -1 in the censoring case, then the average energy consumed by a node is equal to:
\begin{equation}
\label{energy}
E_{av} =\left( p_{H_0}(1 - \pi_{00}) + p_{H_1}(1 - \pi_{01})\right)E \leq E,
\end{equation}

\subsection{Overhead}
The binary consensus algorithm requires nodes to exchange data packets between neighbors during the running time of the algorithm. Specifically, in a conventional binary consensus setting, a node transmits $K$ data packets to direct neighbors if the algorithm takes $K$ time steps to run.

In a censoring-enabled setting, however, a node which decides to censor transmission will not congest the network and therefore will reduce the average number of packets transmitted per node. Given that $C$ is the number of nodes which made a decision to censor transmission in a network, then the average number of censoring nodes at any time can be expressed as:
\begin{equation}
\mathcal{E}(C) = p_{H_0}*\mathcal{E}_0(C) + p_{H_1}*\mathcal{E}_1(C)
\end{equation}

\begin{figure}
\centering
 \includegraphics[width=2.6 in]{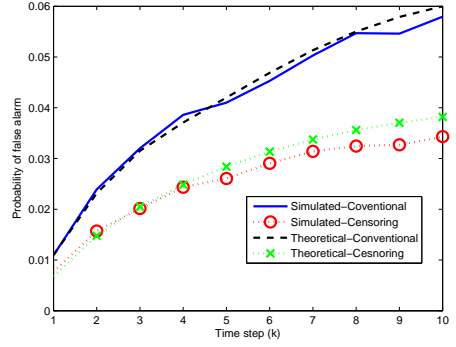}
 \caption{False Alarm Probability: Simulation versus analytical expression - $M=51$, $p=0.8$, $\bar{\gamma}=2 dB$, $\eta=10.3$, $\eta_1=14.6$, $\eta_0=7$, $TB=5$.}
  \label{fig::false_alarm_verify}
\end{figure}

\noindent where $\mathcal{E}(X)$ is the expectation of $X$ and $\mathcal{E}_j(X)$ is the conditional expectation of $X$ given $H_j$. The distribution of $C$ conditioned on $H_j$ is binomial with probability of success equal to $\pi_{0j}$ out of $M$ trials, therefore the expectation of which is given by $\mathcal{E}_j(C) = \pi_{0j}M$. Therefore the average number of censoring nodes is equal to:

\begin{equation}
\mathcal{E}(C) = (p_{H_0} \pi_{00} + p_{H_1} \pi_{01})M
\end{equation}

\noindent and the average number of transmitted messages per node in this case is equal to:
\begin{equation}
\label{overhead}
\begin{split}
\dfrac{(1 - (p_{H_0} \pi_{00} + p_{H_1} \pi_{01}))M \times K}{M} &= \\
(1 - (p_{H_0} \pi_{00} &+ p_{H_1} \pi_{01}))K \leq K
\end{split}
\end{equation}

\subsection*{Validation}
In this subsection, we validate the obtained expressions via numerical simulations. We assume a CRN that consists of $M$ randomly deployed SUs with a uniform distribution. We assume $p=0.8$ and $\bar{\gamma} = 2$ dB. Figures \ref{fig::detection_verify} and \ref{fig::false_alarm_verify} show a comparison between the simulated and analytically obtained values of the probability of detection and false alarms respectively, which clearly shows the soundness of the obtained expressions (\ref{pd_censoring}) and (\ref{pfa_censoring}). The values of the thresholds that were used to produce the two figures are obtained by optimizing over the average error probability, the procedure of which is explained in the following section.  Similar numerical validations were made for equations (\ref{energy}) and (\ref{overhead}) but were omitted due to space limitations..


\section{Numerical Evaluation}
In this section, we compare the performance of our proposed algorithm with its conventional counterpart. We provide numerical evaluation of the obtained expressions in (\ref{pd_conventional}), (\ref{pfa_conventional}), (\ref{pd_censoring}) and (\ref{pfa_censoring}). We numerically optimize the average probability of error over the decision thresholds, i.e., we search for the values of $\eta$, $\eta_1$ and $\eta_0$ which minimize the average probability of error in both the conventional and censoring cases. We compare three different scenarios to observe the performance of our proposed system in various network conditions: 1) Worst-case scenario: poor network connectivity ($p=0.2$) and bad communication channel ($\bar{\gamma} = 2$ dB), 2) Best-case scenario: high network connectivity ($p=0.8$) and good communication channel ($\bar{\gamma} = 4$ dB), and 3) Normal scenario: high network connectivity ($p=0.8$) and bad communication channel ($\bar{\gamma} = 2$ dB).

\subsection{Average Error Probability}
Figure~\ref{fig::error_probability} shows the average error probability of the conventional and censoring systems against $K$ for $M=51$. It is clear that censoring-enabled systems outperform conventional systems for all tested scenarios, which implies that it is always a sane option to employ our proposed algorithm for detection performance with less values of average error probability. In fact, highest performance gains are achieved in worst case scenario ($p=0.2$ and $\bar{\gamma}=2$ dB), where performance gain up to 46.6\% can be attained for $K=1$.

\subsection{Energy Expenditure}
Figure \ref{fig::energy_expenditure} shows the average energy expenditure for all three scenarios with censoring sensors for $E=1$. The average energy expenditure is not significantly affected by the change of the network conditions nor $K$. It is clear though that for all cases, censoring-enabled systems provide energy savings than their conventional counterpart, with an average gain of 48.7\% of energy savings.


\subsection{Network Overhead}
 Figure \ref{fig::overhead} shows the average incurred overhead for conventional and various censoring-enabled cases with a performance gain up to 50\%, which also shows that our proposed protocol proves superior on that front.

\section{Conclusions}
In this paper, we proposed a distributed detection framework for infrastructure-less CRNs which allows SUs to censor transmission. Our proposed system is based on using binary consensus algorithm for data exchange between SUs and therefore does not require the presence of a FC. We derived analytical expressions for performance metrics such as average error probability, energy expenditure and incurred overhead. We validated the obtained expressions via numerical simulations, and we proved that our proposed system significantly outperforms existing non-censoring counterparts in all the previously mentioned performance metrics. Performance gains were established up to 46.6\% in average error probability, energy savings of 48.7\% and up to 50\% savings in incurred transmission overhead.


\bibliographystyle{IEEEbib}
\bibliography{lib}

\end{document}